\documentclass[12pt]{article}

\usepackage{amsmath}
\usepackage{amssymb}
\usepackage{amsthm}
\usepackage{bbm}
\usepackage{algorithm2e}
\usepackage{graphicx}
\usepackage{natbib}
\usepackage[hyphens]{url}
\usepackage{hyperref}
\usepackage[top=1in, bottom=1in, left=1in, right=1in]{geometry}
\usepackage{physics}

\newcommand{\bbR}{\mathbb{R}}
\newcommand{\bbZ}{\mathbb{Z}}
\newcommand{\ra}{\rightarrow}
\newcommand{\one}{\bold{1}}
\newcommand{\ep}{\varepsilon}

\theoremstyle{definition}
\newtheorem{theorem}{Theorem}

\newtheorem{proposition}[theorem]{Proposition}
\newtheorem{corollary}[theorem]{Corollary}

\newtheorem{aside}[theorem]{Aside}

\title{Sampling from a Gaussian distribution conditioned on the level set
  of a piecewise affine, continuous function}

\author{Jesse Windle \\ \url{jesse@bayesfactor.net}}

\begin{document}

\maketitle

\begin{abstract}
  We consider how to use Hamiltonian Monte Carlo to sample from a distribution
  whose log-density is piecewise quadratic, conditioned on the sample lying on
  the level set of a piecewise affine, continuous function.
\end{abstract}

\section{Introduction}

Our motivation comes from modeling monocot (e.g. maize) root growth in a piecewise
linear fashion and inferring the path of a root given a point along its
trajectory.  We defer details of this application to another paper.  Our focus
here is a question that arose when studying that problem.

Suppose that a priori $X \sim N(\mu, I_n)$.  Subsequently, we learn that
$\ell(X) = 0$ where $\ell$ is a piecewise affine, continuous function.  How
does one sample $(X | \ell(X) = 0)$?

The answer to this question extends to the more general problem, which can be
roughly phrased as: how does one sample $(X | \ell(X) = 0)$ when, marginally,
$X$ comes from a density proportional to
\[
\exp \Big[ -\frac{1}{2} x' M_i x + r_i' x + k_i \Big], \text{ if } x \in R_i, \;
i=1, \ldots K
\]
where $R_i \subseteq \bbR^n$ and $R = \cup_{i=1}^K R_i$ is a connected set.
Notice that, in effect, this provides a way to sample any distribution,
approximately, on discretized volumes, conditioned upon an observation lying on
surface embedded in $\bbR^n$, which has been discretized into linear pieces.

We will use Hamiltonian Monte Carlo (HMC) to do this.  To be clear: our
application of interest is the discretized version, not the approximation of
something continuous, since we can use HMC to simulate from a smooth density.
In that case, one must discretize an ODE.  We have, in effect, made a trade for
what must be discretized, sampling exactly from the above discretization by
being able to solve the dynamics of motion exactly.

Part of the aim of this paper is to encourage a physics-based intuition of
particle dynamics for HMC, specifically as it relates to how one deals with
truncations and step changes in densities.  To that end, we have synthesized the
work of \cite{Pakman2014}, who deal with truncations, and
\cite{MohaselAfshar2015}, who deal with steps, and then place that within the
context of a particle moving on a manifold within a higher dimensional space.

Accompanying this paper is a Python package and example Jupyter notebooks
(\cite{ctgauss}) to show how to use the software.  In many modeling scenarios, the
excellent HMC software Stan (\cite{Stan2023}) will work, even with truncations
or for some non-smooth densities (see e.g. \cite{StanFunctions}, section 3.7),
and we encourage you to use that software.  However, to the best of our
knowledge and as of this writing, Stan may be unable to accomodate some use
cases, like arbitrary truncations by hyperplane in multiple dimensions or when
the piecewise affine, continuous function used for conditioning is complex.

\section{Background}

We will be using HMC to generate samples from the distribution above.
Hamiltonian Monte Carlo is inspired by Hamiltonian mechanics.  (\cite{Kibble2004}
is a well-known introductory text to classical mechanics.  I found the HMC notes
by \cite{Vishnoi2021} to be a useful, quick introduction and borrow from them
here.  \cite{Betancourt2017} is a more extensive introduction.  \cite{Neal2011}
is another useful review for some of the key concepts.)

In physical systems that are governed by a time-independent potential, the total
energy can be written as
\[
E = V(r) + \frac{1}{2m} \|m \dot r\|^2,
\]
which is constant.  It is often the case that the positions and momenta of the
system can be written in only a few coordinates, called generalized positions,
$x$ and momenta, $s$ so that $E = H(x, s)$.  One can show using variational
methods, that the evolution of a particle in $xs$-phase space is governed by
\begin{align*}
  \dot x & = \nabla_s H \\
  \dot s & = -\nabla_x H.
\end{align*}
Let $\varphi_\tau$ be the map $(x(\tau), s(\tau)) = \varphi_\tau(x_0, s_0)$ that represents the
state of the system after time $\tau$ and starting at $(x_0, s_0)$.  If one defines
a distribution in phase space so that its density is proportional to
$\exp \{ -H(x,s) \}$, and if one samples random variables $(X, S)$ from this
distribution, then the map of the evolution $\varphi_\tau(X, S)$ will preserve
the distribution, where $\tau$ is any time step.  In other words, the distribution
is stationary under the mapping $\varphi_\tau$.

For the purposes of random variate generation, one can consider
$H(x,s) = V(x) + T(s)$ where $V(x)$ is the negative logarithm of a distribution
of interest and $T(s)$ is the negative logarithm of a proposal distribution.  We
can construct a Markov chain as follows.  Let $X_0$ be an arbitrary
starting position.  Let $X_n = \varphi_\tau(X_{n-1}, S_{n-1})$ where $S_n$ is a
random variable sampled from the distribution derived from $\exp(-T(s))$.  When
$T$ and $V$ satisfy regularity conditions and $\tau$ is an acceptable step size,
$(X_n)_{n=1}^{\infty}$ is a Markov chain that converges to the distribution of interest.

\subsection{Example: Gaussian distribution}
\label{sec:gaussian-example}

Let us consider sampling from a Gaussian random variable using HMC.  This may
seem trivial, since we already known how to sample from a Gaussian distribution
without HMC, but it will be instructive as well as useful for later.  Suppose
\begin{gather*}
  V(x) = \frac{1}{2} x' M x - r'x \\
  T(s) = \frac{1}{2} s' M^{-1} s.
\end{gather*}
The Hamiltonian for the system is then
\[
  H(x, s) = \frac{1}{2} x' M x - r'x + \frac{1}{2} s' M^{-1} s
\]
and the evolution in this phase space is given by
\begin{align*}
\dot x & = \frac{dH}{ds} = M^{-1} s \\
\dot s & = -\frac{dH}{dx} = -M x + r.
\end{align*}
Note that $s$ in this case \emph{is not} the actual momenta, but rather a linear
combination of the particle's momentum.  Together this implies that the motion
of $x$ is governed by
\[
\ddot x + x = \mu 
\]
where $\mu = M^{-1} r$, which has the solution
\begin{equation}
  \label{eqn:equations-of-motion}
x_i(t) = \mu_i + a_i \sin(t) + b_i \cos(t)
\end{equation}
with $b_i = x_i(0) - \mu_i$ and $a_i = \dot x_i(0) = -(M^{-1})_{i\cdot} s_0$.
We can sample $s$ from the correct marginal distribution, $s \sim N(0, M)$, to
generate the Markov chain, or by the transformation $\dot x = -M^{-1}$, we see
that we can also sample $\dot x(0)$ directly as $N(\mu, M^{-1})$ --- yes, we are
saying we can sample $N(0, M^{-1})$ to generate a Markov chain whose stationary
distribution is $N(\mu, M^{-1})$.  If we chose $\tau = \pi/2$, we can even
generate independent and identical samples.  This may seem rather tautological
at the moment, but it will be very useful when dealing with truncated Gaussian
random variables or piecewise quadratic likelihoods.

\subsection{Example: piecewise (or truncated) Gaussian}

To see why where HMC becomes useful, let us now consider a potential defined
piecewise, so that
\begin{equation}
  \label{eqn:canonical-potential}
  V(x) =
  \begin{cases}
    \frac{1}{2} x' M_1 x - r_1' x + k_1, \; \text{ if } f'x + g > 0, \text{
      region 1} \\
    \frac{1}{2} x' M_2 x - r_2' x + k_2, \; \text{ if } f'x + g < 0, \text{
      region 2}
    \end{cases}.
\end{equation}
Previously, the constant term in the polynomial did not matter, but now it does
since that will weight each piece differently.  Also, we could let e.g.
$k_2 = \infty$ to effectively truncate the distribution.

We can now evolve a particle by Hamiltonian dynamics, but we have to be mindful
of the boundary.  The key question is: when does the particle run into the
boundary?  Following \cite{Pakman2014}, let
\begin{align*}
  K(t) & = f' x(t) + g \\
         & = f'a \sin(t) + f' b \cos(t) + (f' \mu + g) \\
         & = u \cos(t + \phi) + h
\end{align*}
where
\[
u = \sqrt{(f'a)^2 + (f'b)^2},
\]
\[
  \tan \phi = \frac{f'a}{f'b}
\]
and
\[
h = f' \mu + g.
\]
The particle hits a boundary at the first time that $K(t) = 0$, which can only
happen if $u > |h|$.  (As a technical matter, when computing $\phi$, one should
use the \texttt{atan2} function, which returns the angle in $[-\pi, \pi]$ so
that $f'a$ and $f'b$ have the correct signs.)

We want to find the time $t > 0$ in which a particle hits a boundary, as it is
going in a direction that would exit the region.  To that end, if $u_j < |h_j|$,
let $\tau = \infty$, otherwise let
\[
  \tau = \min \Big\{ t^{(n)} = \arccos (-h / u) - \phi + \pi n :
  n \in \bbZ, \; t^{(n)} > 0, \; K'(\tau^{(n)}) <  0 \Big\}.
\]
To clarify: first, we have to find positive $\tau$ hence the adjustment by
multiples of $\pi$.  Second, we want to make sure that the particle is exiting
the boundary, hence the condition on the derivative.  If $\tau \geq \tau_{max}$,
$\tau_{max}$ the max step size, then we take a step $\tau_{max}$.  If
$\tau < \tau_{max}$, then we need to adjust the momentum of the particle as it
encounters the boundary and then let it continue on its way.


Suppose the particle starts in region 1 and at time $\tau < \tau_{max}$ it encounters the
boundary.  The intuition of how to approach the problem is as follows.  Suppose
$V(x)$ is piecewise constant so that only $k_1$ and $k_2$ are non-zero.  We can
think of this potential as a plain (region 1) and a plateau (region 2), both of
which are totally flat with a constant force of gravity perpendicular to both
surfaces.  We can relax the problem to make the vertical face of the plateau
slightly tilted, and even smooth.  The key is that the shape of the relaxation
only changes in the direction orthogonal to the boundary.  Under that
assumption, the value that clearly matters is the momentum in this direction.
If there is enough momentum to overcome the potential energy due to gravity, a
particle will ascend the plateau.  If it does not, then it will return to region
1.  As we make the relaxation steeper and steeper, the time in spends on this
portion of the curve decreases, so that in the limit, it will either ascend the
plateau and lose momentum in the direction orthogonal to the boundary, or it
will reflect.

The direction orthogonal to the boundary, pointing in, is $u = f / \|f\|$, hence the velocity
(and momentum since $m=1$) is $v_\perp = u' \dot x(\tau)$.  The associated
kinetic energy is $E_\perp  = \frac{1}{2} v_\perp^2$.  Let
\[
  \Delta V = \lim_{z \ra x(\tau), f'z + g < 0} V(z) -  \lim_{z \ra x(\tau), f'z
    + g>
    0} V(z).
\]
If $E_\perp > \Delta V$, then the particle will move into the other region with
a subsequent reduction in momentum in the $u$ direction of
\[
v_{new} = \sqrt{2 (E_\perp - \Delta V)},
\]
which is the key point of \cite{MohaselAfshar2015}, so that
\[
\dot x(\tau^+) = (I - u u') \dot x(\tau^-) + v_{new} (-u),
\]
(the minus in front of $u$ ensures we are pointing into the next region),
otherwise all of the momentum in question will be reflected;
\[
  \dot x(\tau^+) = (I - uu') \dot x(\tau^-) - v_\perp u,
\]
which we can summarize as
\[
  \dot x(\tau^+) = \dot x(\tau^-) + u
  \begin{cases}
    - 2 v_\perp, \; v_\perp^2 < 2 \Delta V \\
    -v_\perp - \sqrt{v_\perp^2 - 2 \Delta V}, v_\perp^2 \geq 2 \Delta V.
    \end{cases}
  \]

Note, when $k_2 = \infty$, the boundary is effectively a truncation of the
distribution, and you are guaranteed to reflect.

Given this new velocity, we can re-initialize the dynamics, but now in
the appropriate region with initial conditions $(x(\tau), \dot x(\tau^+))$ and
max running time $\tau_{max} - \tau$.

\section{Parameterization of distribution}

First, let us define the distribution and its parameters. In the previous
section, we considered a truncated Gaussian distribution. Here, we consider a
Gaussian, subject to a specific type of piecewise linear contraints. We
can think of this as starting with the marginal distribution
\begin{gather*}
  X \sim \exp(-V(x)) \\
  V(x) = \begin{cases}
    \frac{1}{2} x' M_j x - {r_j}' x + k_j, \; x \in R_j \\
    R_j = \{ x: {F_j} x + g_j \geq 0 \} \\
    j = 1, \ldots, J
  \end{cases},
\end{gather*}
and then wanting to sample $(X | \ell(X) = 0)$ where $\ell$ is a continuous
function that can be written as
\[
\ell(x) = {A_j}' x + y_j, \text{ if } x \in R_j.
\]

That provides a general sense of what we want to do, but we need to restrict the
problem further in order to avoid degenerate cases and provide a succinct
parameterization.  To begin, let
\[
{f_i}' x + g_i \geq 0, \; i = 1, \ldots, n
\]
be a collection of hyperplanes.

Let the n-tuple $s_j \in \{-1, 0, 1\}^n$ define the $j$th region, $R_j$ by
imposing the constraints
\[
s_{ji} \Big[ {f_i}' x + g_i \Big] \geq 0, \; i = 1, \ldots, n.
\]
When $s_{ji} = 0$, then the constraint is not active; if it is non-zero it is
active and its sign determines which half of the half-space to use.  We also
want to keep track of the intended transitions.  To that end, let the magnitude
of $s_{ji}$ indicate the index of the adjoining region and when
$|s_{ji}| = j$, we will assume it represents a hard boundary.  Let $L$
is the $J \times m$ matrix of these $m$-tuples so that $s_{ji} = L_{ji}$ ($L$
for lookup).  If a particle can move from $j$ to $j^*$ and visa versa, then
there should be some constraint $i$ such that $j^* = |L_{ji}| \neq 0$ and
$j = |L_{j^*i}| \neq 0$.

Thus, to recap, the parameterization includes:
\begin{itemize}
\item The potential $V$: that is $(M_j, r_j, k_j)$, where $M_j$
  is a $d \times d$ matrix, $r_j$ is an $\bbR^d$ vector, and $k_j$ a scalar, for
  $j=1, \ldots, J$.
\item $f_i, i = 1, \ldots, m$, vectors in $\bbR^n$ used to construct the
  hyperplane boundaries, or $F$ a $m \times n$ dimensional matrix whose rows are
  the $f_i$.
\item $g_i, i = 1, \ldots, m$, the scalars that control the offsets of the
  hyperplanes, or $g$ an $m$-dimensional vector.
\item $s_j, j = 1, \dots J$, the m-tuples defining the regions, or $L$ a $J
  \times m$ dimensional matrix whose rows are the $m$-tuples.
\item The continuous, piecewise affine function $\ell$: that is $A_j$ and
  $y_j$, where $A_j$ is a $n \times d$ and $y_j$ is a $\bbR^d$
  vector, $j=1, \ldots, J$.
\end{itemize}

\subsection{Restrictions on parameters}

Since we will be using HMC, as a technical matter, it will be helpful to avoid
defining regions so that a face of one region abuts the faces of two adjacent
regions.  This will prevent us from having to search through all the regions to
determine which region the particle is moving into.  We can avoid this by
imposing: for each $j$, for each $i$, there is at most one $j'$ such that
$0 \neq s_{ij} = -s_{j'i}$.  We also want $\ell$ to be continuous.  The user
should define the regions so that the interior of $R$ is connected.

\section{Dynamics}

Let us return to the potential (\ref{eqn:canonical-potential}):
\[
  V(x) =
  \begin{cases}
    V_1(x) := \frac{1}{2} x' M_1 x - r_1' x + k_1, \; \text{ if } f'x + g > 0 \;
    (R_1) \\
    V_2(x) := \frac{1}{2} x' M_2 x - r_2' x + k_2, \; \text{ if } f'x + g < 0
    \; (R_2)
    \end{cases}.
\]
and consider the negative log density of the marginal
distribution of interest.  Our aim is to sample from conditional upon $\ell(x) = 0$ where
\[
  \ell(x) = 
  \begin{cases}
  A_1'x + y_1, \; f'x + g > 0 \; (R_1) \\
  A_2'x + y_2, \; f'x + g < 0 \; (R_2)
  \end{cases}
\]
and $\ell$ is continuous.  There are two issues to address: (i) what equations
should drive the dynamics of the particle, and (ii) what happens when a particle
hits the boundary?

\subsection{Dynamics within a region}

Suppose a particle's dynamics are governed by the unconstrained force that has the
form
\[
F(x) = -Mx + r,
\]
which corresponds to a Gaussian potential with precision $M$ and mean
$\mu = M^{-1}r$, but that we then constrain the motion of a particle to the
plane $A'x = y$.  How do we adjust the force and subsequent potential to
accurately represent the particle's dynamics, and how can we connect it to the
unconstrained potential?

Let $Q_1$ be the unitary column space of $A$ and $Q_2$ be the unitary null
space.  We need to make two observations.  First, the force within the
constrained space is
\[
- Q_2 Q_2' [ Mx - r ].
\]
Second, if we decompose $x$ as $Q_1 Q_1' x + Q_2 Q_2' x$, then the portion
$Q_1 Q_1' x$ is constant.  Specifically, letting $A = Q_1 R$, we see that
$Q_1' x = {R'}^{-1} y =: z_1^*$, so that $Q_1 Q_1' x = Q_1 z_1^*$.

Hence the force of interest is
\begin{align*}
  F(x) & = - Q_2 Q_2' M Q_2 Q_2' x - Q_2 Q_2' M Q_1 Q_1' x + Q_2 Q_2' r \\
  & = - Q_2 (Q_2' M Q_2) Q_2 x + Q_2 Q_2' [ r - M Q_1 z_1^*].
\end{align*}
Letting $\tilde r = [r - M Q_1 z_1^*] = M[\mu - Q_1 z_1^*]$, the prior expression implies
the potential
\[
V(x) = \frac{1}{2} x' Q_2 (Q_2' M Q_2) Q_2 x - (Q_2 Q_2' \tilde r)' x.
\]

To write down the Hamiltonian for the purposes of generating random variates, we
need to chose a kinetic energy term that makes generating the dynamics easy.


From section \ref{eqn:canonical-potential}, we know that $s' M^{-1}s $ is a
good choice, where $M$ is the potential.  This suggests we should choose a
pseudo inverse of $Q_2 (Q_2' M Q_2)^{-1} Q_2'$ when constructing the kinetic
energy term, which we also know is a good choice by solving the problem in lower
dimensional coordinates.  To that end, we want
\[
H(x, s) = \frac{1}{2} x' (Q_2 (Q_2' M Q_2) Q_2') x - (Q_2Q_2' \tilde r)' x +
\frac{1}{2} s' Q_2 (Q_2' M Q_2)^{-1} Q_2 s,
\]
which leads to the dynamics
\[
\ddot x + Q_2 Q_2' x = Q_2 (Q_2' M Q_2)^{-1} Q_2' \tilde r.
\]
This has homogeneous and particular solutions
\[
x_h(t) = a \sin(t) + b \cos(t)
\]
and
\[
x_p(t) = Q_2(Q_2' M Q_2)^{-1} Q_2' \tilde r + Q_1 z_1^*.
\]
Then the solution is $x = x_h + x_p$ where
\[
b = x(0) - x_p
\]
and
\[
a = \dot x(0),
\]
with the requirement that $Q_1' \dot x(0) = 0$.  To generate a sample from the
marginal of $s$, we can sample $s \sim N(0, Q_2 (Q_2' M Q_2) Q_2')$, which is
like endowing $\dot x(0)$ with initial velocity from $\dot x(0) \sim N(0, Q_2 (Q_2' M
Q_2)^{-1} Q_2')$.

\subsection{Change in dynamics between regions}
\label{sec:change-in-dynamics}

To begin, let us consider a particle that is traveling free of any force, but
whose movement is restricted to $R_1$ and $R_2$.  Understanding the velocity of
the particle, and how we can decompose it is key.  Keep in mind that, if the
particle is moving from region 1 to region 2, then the particle's velocity must
be in the null space of $A_1'$, and when it crosses the boundary, it will be
in the null space of $A_2'$.  It stands to reason that the components of the
velocity that are shared between each region will not change.

Let us assume we can find a semi-orthonormal matrix $U_0$ and
unit vectors $u_i, i=1,2$ so that $u_i$ is orthogonal to $U_0$ and the span of
$u_i$ and $U_0$ is the null space of $A_i'$, $i=1, 2$.  Further, assume there is
another semi-orthonormal matrix $U_c$ such that $U_c' U_0 = 0$ and
$u_i'U_c = 0$, $i=1,2$.  The justification is given in Proposition
\ref{thm:null-space-decomposition} and Corollary \ref{thm:subspace-decomposition} of
Appendix \ref{sec:subspace-decomp}.  To orient $u_1$ and $u_2$, assume they are
pointing into each region so that $f' u_1 > 0$ and $f' u_2 < 0$.

How does a particle move from $R_1$ into $R_2$?  Our approach is to set up a
hard boundary so that a particle traveling in $R_1$ has an inelastic collision
that bounces it directly into $R_2$.  As before let $\dot x(\tau^-)$ be the
particle's velocity when it encounters the boundary.  There will be an
instantaneous change in velocity to $\dot x(\tau^+)$ where
\begin{equation}
  \label{eqn:glance-through-boundary}
x(\tau^+) = (I -  u_\perp u_\perp') x(\tau^-) - u_\perp u_\perp' x(\tau^-).
\end{equation}
where
\[
u_{\perp} = \frac{u_2 + u_1}{\|u_2 + u_1\|}.
\]
It will be helpful to define
\[
u_{\|} = \frac{u_2 - u_1}{\|u_2 - u_1\|}.
\]
as well.  Geometrically speaking, $u_{\perp}$ is the vector that bisects $u_2$
and $u_1$ and $u_{\|}$ is the vector the bisects $u_2$ and $-u_1$.  Because
$u_1$ and $u_2$ are orthogonal to $U_0$ and $U_c$, we have that the
concatenation of $U = [u_\perp \; u_\| \; U_0 \; U_c]$ (or
$U =[u_\perp \; U_0 \; U_c]$ when $u_1 = -u_2$) is an orthonormal matrix and
$U'U = I$, which tells us we can also decompose
(\ref{eqn:glance-through-boundary}) as
\begin{displaymath}
\dot x(\tau^+) = - u_{\perp} u_{\perp}' \dot x(\tau^-) + u_{\|} u_{\|}' \dot x(\tau^-) +
U_0U_0' \dot x(\tau^-) + \underbrace{U_c U_c' \dot x(\tau^-)}_{=0}.
\end{displaymath}
The last term is eliminated because $\dot x(\tau^-)$ is in the null space of
$A_1'$.

To see that the particle does indeed take its trajectory in $u_2$, we can
decompose $\dot x(t^-)$ as $-\alpha u_1 + w$ where $w \in U_0$ and $\alpha > 0$,
since the particle is leaving the region.  The subsequent trajectory is
\begin{gather*}
  (u_2 + u_1) \frac{(u_2+u_1)'}{\|u_2+u_1\|^2} \alpha u_1 -
  (u_2 - u_1) \frac{(u_2-u_1)'}{\|u_2 - u_1\|^2} \alpha u_1 + w \\
  = \alpha (u_2 + u_1) \frac{(u_2'u_1 + 1)}{2(1+u_2'u_1)} -
  \alpha (u_2 - u_1)\frac{u_2'u_1 - 1}{2(1-u_2'u_1)} + w
  \\
  = \alpha u_2 + w.
\end{gather*}

In other words, we can project the velocity of the particle into $-u_1$ (the
direction of exit) and then transfer all of that velocity to $u_2$ and keep the
remaining components of velocity the same.  Since the particle is traveling in
$R_1$ to begin, with, this is equivalent to removing the velocity due to $u_1$
and then redirecting it to $u_2$.  That is
\begin{align*}
  \dot x(\tau^+) & = (I - u_1 u_1') \dot x(\tau^-) + u_2 (-u_1)' \dot x(\tau^-).
\end{align*}

When the potential is continuous at $x(\tau)$, then the total energy will not
change, since momentum is conserved.  In order to tackle potentials that are
defined piecewise, we must consider how the momentum of the particle changes.

Instead of considering a flat plain and plateau, consider a plain that is tilted
in the direction of $u_1$ and a plateau that is tilted in the direction of
$u_2$.  The sheer face of the plateau is predominantly in the direction of
$u_\perp$, so that it is nearly vertical.  Now imagine that we relax the sheer
face so that it is steep and that there is a force of gravity in the direction
of $-u_\perp$ and only in the portion of the relaxation.  In order for a
particle to ascend to the plateau, we must have the momentum in $u_1$, $v_1 =
u_1' \dot x(\tau^-)$, and its
corresponding energy exceed the potential energy we must climb:
\[
E_{u_1} := \frac{1}{2} (\dot x(\tau^-)' u_1)^2 > \Delta V =: V_2(x(\tau^-)) - V_1(x(\tau^-)).
\]
The drop in energy is then
\[
E_{new} = E_{u_1} - \Delta V
\]
and the subsequent momentum in the $u_2$ direction is
\[
\sqrt{2 E_{new}} u_2.
\]
In terms of $v_1$, we have
\[
\sqrt{2 E_{new}} = |v_1| \sqrt{1 - \frac{2 \Delta V}{v_1^2}} = -v_1 \sqrt{1 - \frac{2 \Delta V}{v_1^2}},
\]
which clarifies that we lose a fractional portion of the speed in the $u_1$
direction as it relates to the change in potential energy.  If there is not
enough energy to ascend the potential, then the momentum in the direction of
$u_1$ is effectively reversed.  That is
\[
  x(\tau^+) = (I - u_1 u_1') \dot x(\tau^-) +
  \begin{cases}
    u_2 \sqrt{2 E_{new}}, \; \text{ if } E_{new} \geq 0 \\
    -u_1 (u_1' \dot x(\tau^-)), \; \text{ if } E_{new} < 0
  \end{cases}.
\]

\begin{aside}
  Given the above equation, there is a useful computational trick to note.
  Assume that the particle is traveling in region 1.  We can make an artificial
  boundary --- that has no effect --- by assuming that $\Delta V = 0$ and
  $u_2 = - u_1$ where $u_1$ is a vector in the null space of $A_1'$.  Similarly,
  we can force there to be a hard boundary by taking $\Delta V = 0$ and
  $u_2 = u_1$ where $u_1$ is a vector in the null space of $A_1'$.  These
  interpretations are useful when vectorizing the dynamics for many particles,
  since algorithm above can then be used for all potential scenarios: no
  boundary, a hard boundary, or a transition boundary.
\end{aside}

\section{Examples}

\subsection{One-norm}

Here we consider sampling a Gaussian, given that the sample must live on a level
set of the one-norm.  Specifically, we want to sample
\[
X \sim \Big( N(0, \Sigma) \; | \: \|X\|_1 = 1 \Big).
\]
Figure \ref{fig:onenorm} shows these samples for $\Sigma = I_3$ and $\Sigma$
diagonal with diagonal elements $(0.1, 10., 10.)$.

\begin{figure}[h]
  \includegraphics[width=6in]{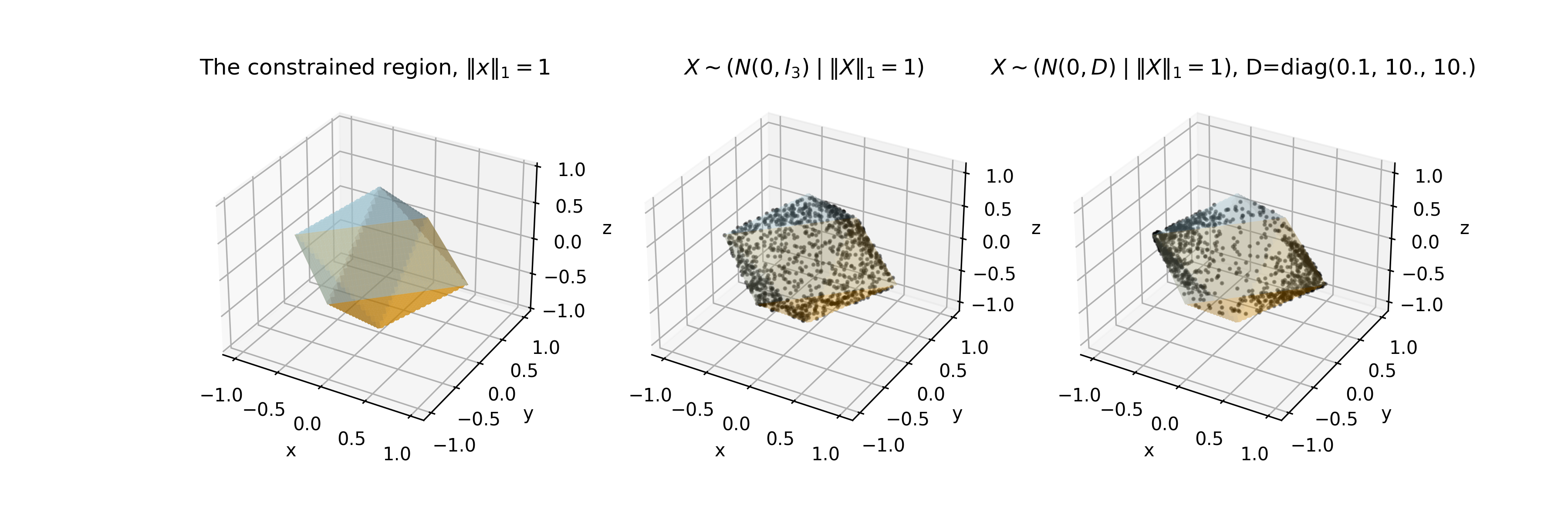}
  \caption{Sampling $X \sim N(0, \Sigma)$ given $\|X\| = 1$.  The left image is
    just the constraint.  The middle image shows a sample for $\Sigma = I_3$.
    The right image shows a sample for $\Sigma$ diagonal with diagonal elements
    $(10, 0.1, 0.1)$.}
  \label{fig:onenorm}
\end{figure}

\subsection{Polygonal top}

Extending the one norm example, here we have constructed a regular polygon in
the xy-plane and then connected its vertices to the z-axis above and below the
xy-plane to create a top.  Figure \ref{fig:ntop} shows the distribution of
$N(0, D)$ restricted to a six-sided top for various covariance structures $D$.
In this case

\begin{figure}[h]
  \includegraphics[width=6in]{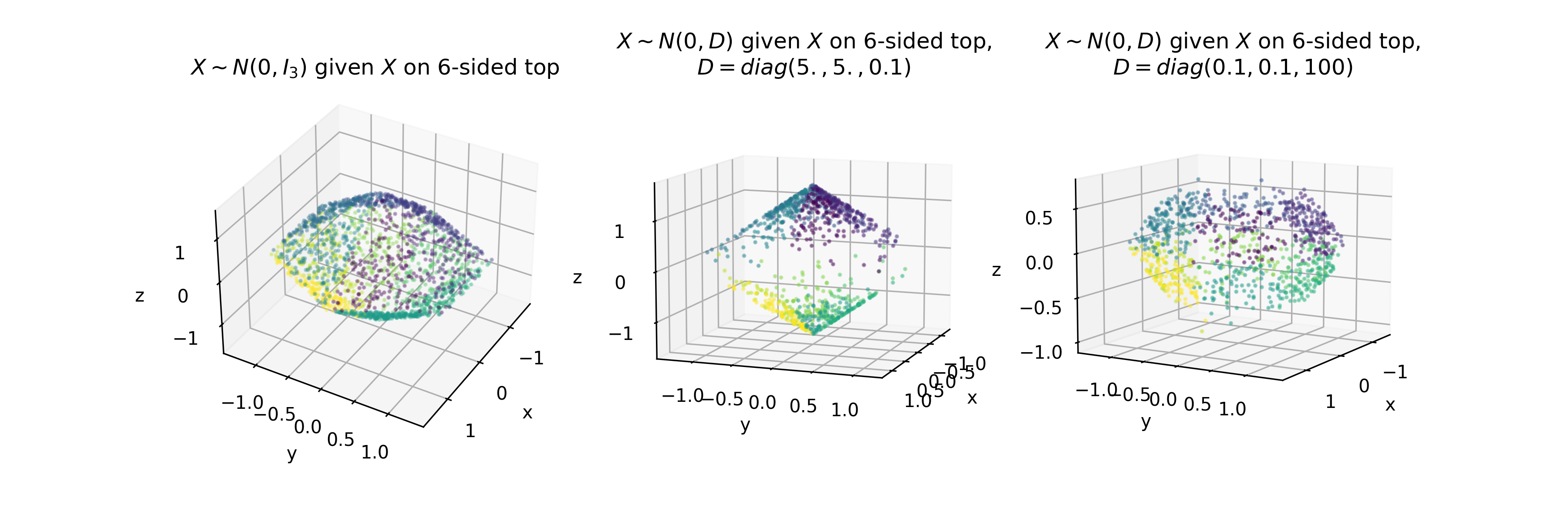}
  \caption{Sampling $X \sim N(0, D)$ given that $X$ lies on a 6-sided top for
    various covariance structures.  The colors represent different regions in
    our partition of $\ell$.}
  \label{fig:ntop}
\end{figure}

\subsection{Positive part}

Here we consider sampling on the subspace that arises when using the $(\cdot)^+$
function.  Our motivation comes from modeling root growth in a piecewise linear
fashion.  We consider a root that has $K=3$ kinks in it.  The depth of the
root after traveling a distance of $r > 0$ in the x-direction is
\[
\ell(\Delta x) = \sum_{i=1}^{K+1} \Delta m_i (r - \sum_{j=1}^i \Delta x_i)^+
\]
where $\Delta m_i, \i=1, \ldots, K+1$ is taken as a parameter here that
determines how the slope of the root changes.  The variable we will be sampling
is $\Delta x_i \geq 0, i = 1, \ldots, 3$.  The constraint $\ell(\Delta x) = d$
can be described by the three regions shown in Figure
\ref{fig:positive-part-regions}.

\begin{figure}[h]
  \includegraphics[width=6in]{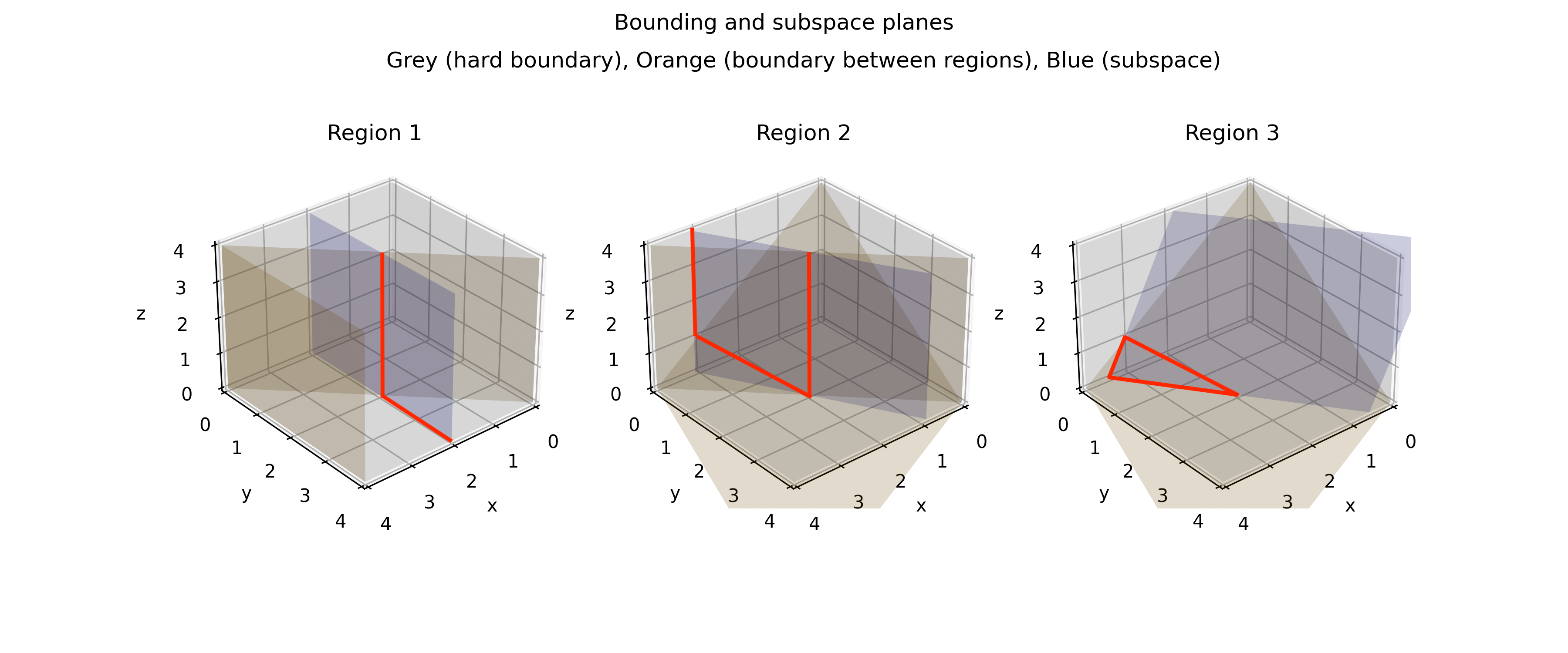}
  \caption{The regions used to define $\ell(\Delta x) = d$.}
  \label{fig:positive-part-regions}
\end{figure}

Figure \ref{fig:positive-part-samples} shows samples of
\[
\Delta X \sim \Big( N(\one_3, I_3) \; \Big| \; \ell(\Delta X) = d \Big).
\]

\begin{figure}[h]
  \begin{center}
    \includegraphics[width=5in]{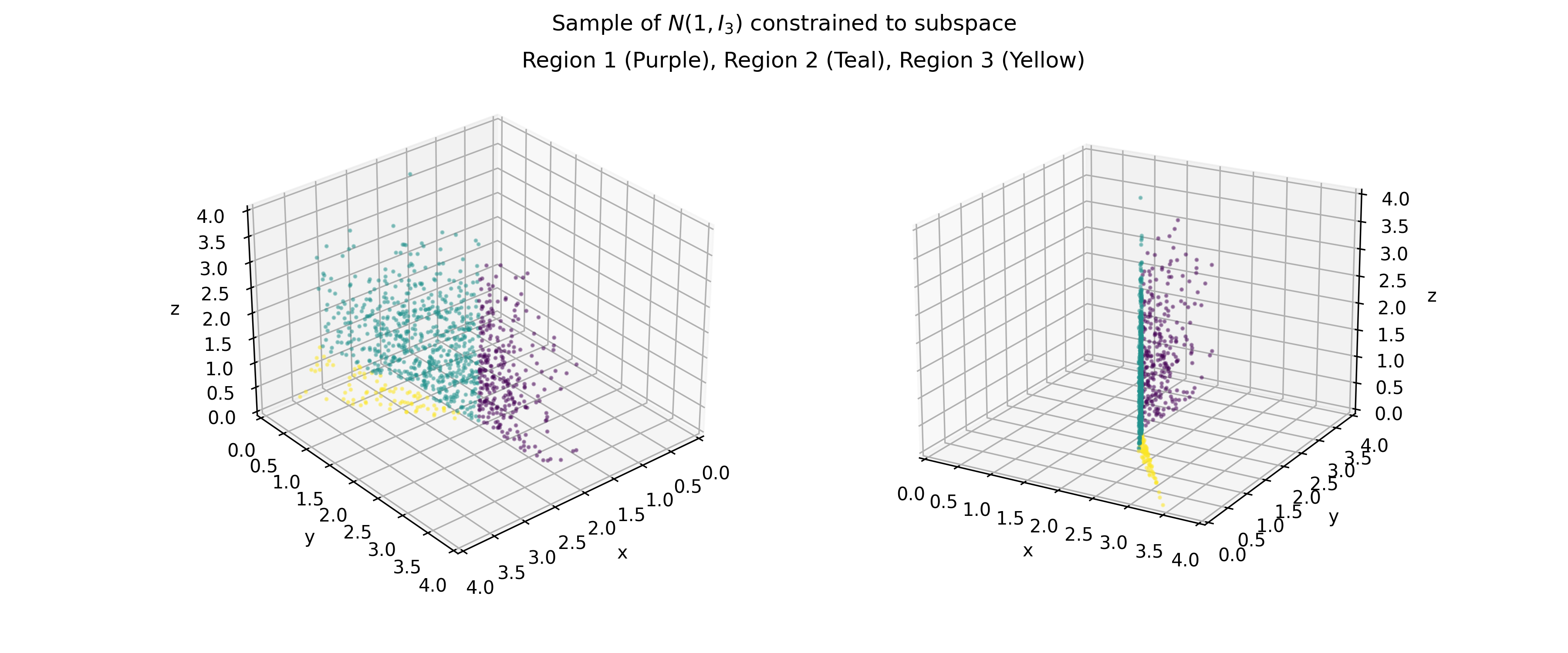}
  \end{center}
  \caption{Samples of $\Delta X \sim ( N(\one_3, I_3) \; | \; \ell(\Delta X) = d)$}
  \label{fig:positive-part-samples}
\end{figure}

\section{Conclusion}

Here we have shown how to start with a distribution that is piecewise quadratic
in the logarithm and sample from that distribution, conditional upon the sample
lying on the the level set of a piecewise affine, continuous function using
Hamiltonian Monte Carlo.

This relies heavily on knowing the exact dynamics of a particle under a
quadratic potential, which can then easily accomodate truncations and step
changes in densities.  To that end, we have synthesized the work of
\cite{Pakman2014}, who deal with truncations, and \cite{MohaselAfshar2015}, who
deal with steps, and then placed that within the context of a particle moving on
a manifold within a higher dimensional space.

We have written a Python package implementing these methods (\cite{ctgauss}).
However, for many use cases, slightly relaxing the requirement that the sample
come from a lower dimensional subspace enables the use of the HMC software Stan,
even with univariate truncations and for some non-smooth densities.  We
encourage that route, if possible.  If one needs to consider complex truncations
or piecewise affine conditioning conditions, then our software may be useful.

\bibliographystyle{abbrvnat}
\bibliography{truncated-constrained-gaussian}{}

\appendix

\section{Intersection of subspaces}
\label{sec:subspace-decomp}

Let
\begin{equation}
  \label{eqn:subspace-constraint}
  \ell(x) = 
  \begin{cases}
  A_1'x + y_1, \; f'x + g > 0 \; (R_1) \\
  A_2'x + y_2, \; f'x + g < 0 \; (R_2)
  \end{cases}
\end{equation}
where $f$ is an $n$-dimensional vector, and $A_1$ and $A_2$ are $n \times d$
matrices.  We want to show that the null spaces of $A_1'$ and $A_2'$ share a
common basis $U_0$ and are spanned by $U_0$ and single vectors $u_1$ and $u_2$,
respectively.

We can find a common basis easily --- just solve $ [A_1 \; A_2]' v = 0.  $ But that
does not guarantee that there are single vectors $u_1$ and $u_2$ that suffice
for spanning the respective null spaces along with the common elements.  The
continuity of $\ell$ is critical for that purpose.

\begin{proposition}
 \label{thm:null-space-decomposition}
 Suppose $\ell$ is continuous and $f$ is not contained in either the column
 space of $A_1$ or $A_2$ and $A_1$ and $A_2$ are full rank.  Then there exists a
 semi-orthonormal matrix $U_0$ and unit vectors $u_1$ and $u_2$, such that $u_i$ and the
 columns of $U_0$ span the null space of $A_i'$ and $u_i' U_0 = 0$, $i=1,2$.
\end{proposition}

\begin{proof}
  Make the requisite assumptions and construct $C_1 = [A_1 \; f]$.  Since $A_1$ is
  full rank and $f$ is not in the column space of $A_1$, we know $C_1$ is full
  rank.  Let $Q_1 R_1$ be the QR decomposition of $C$ and let $Q_0$ be the
  orthogonal complement to $Q_1$ so $Q = [Q_1 \; Q_0]$ is unitary.  Instead of
  working in $x$-space, consider the transformation $x = Qz$ which is
  $Q_1 z_1 + Q_0 z_2$ after decomposing $z$ into two components.  The contraint
  $\ell(x) = 0$ for region 1 becomes
  \[
    A_1' Q_1 z_1 + y_1 = 0
  \]
  where $z_2$ is free.  Similarly, at the boundary $f'x + g = 0$, we have
  $f'Q_1z_1 + g = 0$ and $z_2$ is free.  Hence, $z_1$ must satisfy
  \begin{align*}
    [A_1 \; f]' Q_1 z_1 + b_1 & = 0 \\
    R_1' z_1 + b_1 & = 0.
  \end{align*}
  where $b_1' = [y_1' \; g]$.  The continuity at the boundary requires that
  \[
    A_2' Q z + y_2 = A_2' Q_1 z_1 + A_2' Q_0 z_2 + y_2 = A_2' Q_0 z_2 + (A_2'
    Q_1 z_1 + y_2) = 0
  \]
  for all $z_2$.  Since this must hold for all $z_2 \in \bbR^{n-d-1}$, we must
  have both
  \begin{gather}
    A_2' Q_0 = 0, \label{eqn:cont-cond-1} \\
    A_2' Q_1 z_1 + y_2 = 0. \label{eqn:cont-cond-2}
  \end{gather}
  Equation (\ref{eqn:cont-cond-2}) is an algebraic relationship that must be
  satisfied.  Equation (\ref{eqn:cont-cond-1}) is more critical and it says that
  $A_2'$ must have a null space that includes the span of $Q_0$.  (Alternatively,
  $A_2$ must have columns that reside in the span of $Q_1$, that is, $A_2$ can
  be constructed using the columns of $A_1$ and $f$.)

  Lastly, how do we construct $u_1$, $u_2$ and $U_0$?  Equation
  (\ref{eqn:cont-cond-1}) ensures that $A_2$ is orthogonal to $Q_0$, hence
  $U_0 = Q_0$ characterizes the shared basis of the null spaces of $A_1'$ and
  $A_2'$.  The QR decomposition of $[A_1 \; f]$ motivates the rest.  In
  particular, note that the $(d+1)$th column of $Q_1$, call it $u_1$ is in the null
  space of $A_1'$.  Similarly, let $[A_2 \; f] = Q_2 R_2$ be the analogous QR
  decomposition for $A_2$.  Then $u_2$, the $(d+1)$th column of $Q_2$, will be in
  the null space of $A_2$.
  
\end{proof}

\begin{corollary}
  \label{thm:subspace-decomposition}
  Given the same suppositions as in Proposition
  \ref{thm:null-space-decomposition}, there are unit vectors $u_1$ and $u_2$ and
  semi-orthonormal matrices $U_0$ and $U_c$ such that
  \begin{enumerate}
    \item The span of $U_0$, $u_1$, $u_2$, and $U_c$ is all of $\bbR^n$.
    \item $U_0' U_c = 0$ and $u_i' U_0 = u_i' U_c = 0$, $i = 1, 2$.
    \item The span of $u_i$ and the columns of $U_0$ span the null space of
      $A_i'$ and $u_i'U_0 = 0$, i=1, 2.
    \item Further, if $A_1$ and $A_2$ have the same column span, then
      \begin{enumerate}
        \item $u_1 = \pm u_2$;
        \item $[U_0 \; u_1 \; U_c]$ is an orthonormal matrix.
        \end{enumerate}
    \item if $A_1$ and $A_2$ do not have the same column span, then
      \begin{enumerate}
      \item $U_c$ is $n \times (d-1)$;
      \item $[U_0 \; u_b \; u_\perp \; U_c]$ --- where $u_b = (u_1 - u_2) / \|u_1 - u_2\|$ and
        $u_\perp = (u_1 + u_2) / \|u_1 + u_2\|$ --- is an orthonormal matrix.
      \end{enumerate}
    \end{enumerate}
  \end{corollary}

  \begin{proof}
  By Proposition (\ref{thm:null-space-decomposition}), we have computed $U_0$,
  $u_1$ and $u_2$.  Let $U_c$ be the matrix representation of an orthonormal basis of
  the null space of $[U_0 \; u_1 \; u_2]'$.  This decomposition immediately
  yields (1) - (5).
    \end{proof}

\section{Distribution of a truncated Gaussian conditioned on a hyperplane}
\label{sec:linear-equality-constrained-gaussian}

Here we recall some basic facts about a Gaussian distribution and truncated
Gaussian distribution, conditional on the sample lying in some plane.

Suppose that $x \sim N(\mu, \Sigma)$ and we want to find the distribution of
$(x | A' x = y)$, when $A$ is full rank.  To find the conditional distribution,
we will transform the problem to a different vector of random variables.

Let $Q = [Q_1 \; Q_2]$ be a unitary matrix such that $Q_1$ spans the column space
of $A$ and $Q_2$ spans the null space.  We want to work with $x = Q z$.  Whether
it be by completing the square or using the properties of Gaussian random
variables (and it will matter that we can complete the square later), the
distribution of $z$ is
\begin{align*}
  z & \sim N(\tilde \mu, \tilde \Sigma) \\
  \tilde \mu & = Q' \mu \\
  \tilde \Sigma & = Q' \Sigma Q.
\end{align*}
And, again, by either completing the square or using the properties of Gaussian
random variables, we can decompose the distribution of $z$ as $p(z_2|z_1)p(z_1)$
where $\Omega = \tilde \Sigma^{-1}$,
\begin{align*}
  (z_2|z_1) & \sim N(\tilde m(z_1), \tilde V) \\
  \tilde m(z_1) & = \tilde \mu_2 + A (z_1 - \tilde \mu_1) \\
  A & = \tilde \Sigma_{21} \tilde \Sigma_{11}^{-1} = \Omega_{22}^{-1} \Omega_{21} \\
  \tilde V & = \tilde \Sigma_{22} - A \tilde \Sigma_{11} A' = \Omega_{22}^{-1}.
\end{align*}
We want to sample $(z | A' Q z = y)$, which is to say $(z | A' Q_1 z_1 = y)$.
In this case, $z_1$ is known with certainty and is $z_1 = (A'Q_1)^{-1} y$ and we
can sample $z_2$ from the conditional distribution above.  We can return to the
original random vector $x$ by $x = Qz$ so that
\begin{align*}
  (x | A'x = y) & \sim N(m(z_1^*), V) \\
  m(z_1^*) & = Q_1 z_1^* + Q_2 (\tilde \mu_2 + A ( z_1^* - \tilde \mu_1)) \\
  V & = Q_2 \tilde V Q_2' \\
  z_1^* & = (A'Q_1)^{-1} y.
\end{align*}

\section{Pseudocode}

\SetKwComment{Comment}{/* }{ */}

\RestyleAlgo{ruled}

\begin{algorithm}
  \caption{ODEParam}
  \label{alg:ODEParam}
\KwData{$M$, $r$, $A$, $y$, \text{mean}=False}
\KwResult{$x_p$, $Q$, $S$}
$n \gets \texttt{nrow}(M)$ \;
$d \gets \texttt{ncol}(A)$ \;
$Q, R \gets \texttt{QR}(A, \texttt{mode='complete'})$ \;
$Q_1 \gets Q[{:},0{:}d]$ \;
$Q_2 \gets Q[{:},d{:}]$ \;
$R_1 \gets R[0{:}d, {:}]$ \;
$\Omega_{22} \gets Q_2' M Q_2$ \Comment*[l]{$\Omega_{22} = UU' \implies Q_2 \Omega_{22}^{-1} Q_2 = SS'$ where
  $S = Q_2 {U'}^{-1}$}
$U \gets \texttt{cholesky($\Omega_{22}$, lower=False)}$ \;
$S' \gets \texttt{solve\_triangular($U$, $Q_2'$, trans=$1$, lower=False)}$
\Comment*[r]{$U' S' = Q_2'$}
$z_1 \gets \texttt{solve\_triangular($R_1$, $-y$, trans=$1$, lower=False)}$
\Comment*[r]{$R' z_1 + y = 0$}
$x_1 \gets Q_1 z_1$ \;
\eIf{\texttt{mean}}{
$\tilde r \gets M (r - x_1)$ \;
}{
$\tilde r \gets r - M x_1$ \;
}
$x_p \gets S (S' \tilde r) + x_1$ \;
\end{algorithm}

\begin{algorithm}
  \caption{IsotropicODEParam: using $M = \phi I_n$ and mean $r = \mu$ in
    Alg \label{alg:ODEParamIso}}
\KwData{$\phi$, $\mu$, $A$, $y$}
\KwResult{$x_p$, $Q$, $S$}
$n \gets \texttt{ncol}(\mu)$ \;
$d \gets \texttt{nrow}(A)$ \;
$Q, R \gets \texttt{QR}(A, \texttt{mode='complete'})$ \;
$Q_1 \gets Q[:,0:d]$ \;
$Q_2 \gets Q[:,d:]$ \;
$S \gets Q_2$ \;
$z_1 \gets \texttt{solve\_triangular($R$, $-y$, trans=$1$, lower=False)}$
\Comment*[r]{$R' z_1 + y = 0$}
$x_1 \gets Q_1 z_1$ \;
$\tilde r \gets \phi (\mu - x_1)$ \;
$x_p \gets S (S' \tilde r) + x_1$ \;
\end{algorithm}

\begin{algorithm}
  \caption{GetODEParamForRegion}
  \label{alg:CachedODEParam}
  \KwData{$j$, \texttt{cache}, $M$, $r$, $A$, $y$, \texttt{mean}=False}
  \KwResult{$x_p$, $Q$, $S$}
  \If{\texttt{cache}$[j]$ is Null}{
    \texttt{cache}[j] = ODEParam($M[j]$, $r[j]$, $A[j]$, $y[j]$, \texttt{mean})
  }
  $(x_p, Q, S) \gets cache[j]$\;
\end{algorithm}

\begin{algorithm}
  \caption{GetBoundariesForRegion}
  \KwData{$j$, $F$, $g$, $L$}
  \KwResult{$F_j, g_j, L_j$}
  $active \gets L[j,{:}] \neq 0$ \;
  $signs \gets \texttt{sign}(L[j,active])$ \;
  $F_j \gets F[active, {:}] \odot signs$ \;
  $g_j \gets g[active] \odot signs$ \;
  $L_j \gets L[j,active]$ \;
  \end{algorithm}

\begin{algorithm}
  \caption{GetPotentialForRegion}
    \KwData{$j$, $M$, $r$, $k$}
  $V \gets x' M[j] x + x' r[j] + k[j]$ \;
\end{algorithm}

\begin{algorithm}
  \caption{ContinuityCheck}
  \KwData{$f$, $g$, $A_1$, $A_2$, $y_1$, $y_2$, $tol$}
  \KwResult{$ok$}
  $B_1 \gets [A_1 \; f]$\;
  $d \gets \texttt{ncol}(B)$\;
  $(Q, R) \gets \texttt{QR}(B_1, \texttt{mode='complete'})$\;
  $Q_1 \gets Q[{:}, 0{:}d]$ \;
  $Q_0 \gets Q[{:}, d{:}]$ \;
  $R_1 = R[0{:}d, {:}]$\;
  $z_1 = \texttt{solve\_triangular}(R_1, -y_1, \texttt{trans=1},
  \texttt{lower=False})$
  \Comment*[r]{$R_1' z_1 + y_1 = 0$}
  $e_1 \gets \| A_2' Q_1 z_1 + y_2\|$ \Comment*[r]{$d \times 1$}
  $e_2 \gets \| A_2' Q_0\| $ \Comment*[r]{$d \times (n-d)$}
  $ok \gets (e_1 < tol)$ and $(e_2 < tol)$ \;
\end{algorithm}

\begin{algorithm}
  \caption{ODECoef}
  \KwData{$x_p$, $S$, $x_0$, $\dot x_0=\texttt{Null}$}
  \KwResult{$x_p$, $S$}
  $b \gets x_0 - x_p$ \;
\eIf{$\dot x_0$}{
  $a \gets \dot x_0$
  }{
    $\ep \sim N(0, I_{n-d})$ \;
    $a \gets S \ep$ \;
  }
\end{algorithm}

\begin{algorithm}
  \caption{EvolveToBoundary}\label{alg:two}
  \KwData{$t_{max}$, $x_p$, $a$, $b$, $F$, $g$, $L$, $j$}
  \KwResult{$\tau^*$, $j^*$, $f$, $x$, $\dot x$}
  $Fa \gets F * a$ \;
  $Fb \gets F * b$ \;
  $\phi \gets \texttt{arctan2}(-Fa, Fb)$ \;
  $h \gets F \; \mu + g$ \;
  $u \gets \sqrt{(Fa)^2 + (Fb)^2}$ \;
  $\tau \gets [2 * t_{max} + 2 \pi] * J$ \Comment*[r]{We have to be careful with $\tau$}
  $\tau \gets \texttt{arccos}(-h/u, \texttt{out}=\tau, \texttt{where}=u \geq |h|)$ \;
  $\tau \gets \texttt{where}(\tau < 0, \tau + \pi, \tau)$ \;
  $\tau \gets\texttt{where}(\tau > \pi, \tau - \pi, \tau)$ \;
  $dK \gets Fa \odot \cos(\tau) - Fb \odot \sin(\tau)$ \;
  $\tau \gets \texttt{where}(dK > 0, \tau + \pi, \tau)$ \;
  $\tau^* = t_{max}$, $j^* = j$, 
  $f = a$ \Comment*[r]{Default: does not hit boundary}
  \If{$\texttt{any}(\tau \leq t_{max})$}{
    \Comment*[r]{We need a custom argmin to deal with hitting a corner}
    $k^* \gets \texttt{argmin}(\tau)$ \;
    $j^* \gets |L[j,k^*]|$ \Comment*[r]{Need to subtract 1 if using C-indexing}
    $\tau^* \gets \tau[k^*]$ \;
    $f = F[k^*]$ \;
  }
  $x = x_p + a \sin(\tau^*) + b \cos(\tau^*)$ \;
  $\dot x = a \cos(\tau^*) - b \sin(\tau^*)$ \;
\end{algorithm}


\begin{algorithm}
  \caption{WallDynamics}
  \KwData{$x$, $\dot x$, $u_1$}
  \KwResult{$\dot x_{new}$}
  $v_1 \gets u_1' \dot x$ \;
  $\dot x_{new} \gets \dot x - 2 v_1 u_1$ \;
\end{algorithm}

\begin{algorithm}
  \caption{BoundaryDynamics}
  \KwData{$x$, $\dot x$, $j_1$, $j_2$, $u_1$, $u_2$,
    $V_1$, $V_2$, \texttt{check}=False}
  \KwResult{$\dot x_{new}$, $j_{new}$}
  $v_1 \gets u_1' \dot x$ \;
  $E_{v_1} \gets \frac{1}{2} v_1^2$ \;
  $\Delta V \gets V_2 - V_1$\;
  $\dot x_{new} \gets \dot x - v_1 u_1$ \;
  \eIf{$E_{v_1} < \Delta V$}{
    $j_{new} = j_1$ \;
    $\dot x_{new} \gets \dot x_{new} - v_1 u_1$ \;
  }{
    $j_{new} = j_2$ \;
    $\dot x_{new} \gets \dot x_{new} + \sqrt{2 (E_{u_1} - \Delta V)} u_2 $ \;
  }
\end{algorithm}

\begin{algorithm}
  \caption{resid}
  \KwData{$f$, $d$, $Q$}
  \KwResult{$x$}
  $Q_1 = Q[{:},0{:}d]$ \;
  $x \gets f - Q_1 Q_1' f$\;
  \end{algorithm}

\begin{algorithm}
  \caption{Evolve1}
  \KwData{$t_{max}$, $j$, $x_0$, $\dot x_0$, $M$, $r$, $A$, $y$, $L$, $F$, $g$, \texttt{mean}}
  \KwResult{$x$, $\dot x$, $\tau^*$, $j^*$}
  $(x_p, Q, S) \gets $GetODEParam($j$, \texttt{cache}, $M$, $r$, $A$, $y$,
  \texttt{mean}=False) \;
  $a, b \gets$ ODECoef($x_p$, $S$, $x_0$, $\dot x_0$) \;
  $F_j, g_j, L_j \gets $ GetBoundariesForRegion \;
  $\tau^*, j^*, f, x, \dot x \gets$ EvolveToBoundary($\tau$, $x_p$, $a$, $b$, $F_j$,
  $g_j$, $L_j$, $j$) \;
  \If{$\tau^* > 0$}{
    $u_1 \gets \texttt{resid}(f, d, Q)$ \Comment*[r]{$u_1 = Q[{:}, -1]$ from QR of $[A[j] \; f]$}
    \eIf{$\texttt{hard\_boundary} \gets j^* = j$}{
      $\dot x \gets$ WallDynamics($x$, $\dot x$, $u_1$) \;
    }{
      $(x_p, Q, S) \gets $GetODEParam($j^*$, \texttt{cache}, $M$, $r$, $A$, $y$, \texttt{mean}=False) \;
      $V_1 \gets V_j(x)$ \;
      $V_2 \gets V_{j^*}(x)$ \;
      $u_2 \gets \texttt{resid}(f, d, Q)$ \Comment*[r]{$u_1 = Q[{:}, -1]$ from QR of $[A[j^*] \; f]$}
      $\dot x \gets$ BoundaryDynamics($\dot x$, $u_1$, $u_2$, $V_1$, $V_2$) \;
    }
  }
\end{algorithm}

\begin{algorithm}
  \caption{Evolve2}
  \KwData{$t_{max}$, $j$, $x_0$, $\dot x_0$, $M$, $r$, $A$, $y$, $L$, $F$, $g$, \texttt{mean}}
  \KwResult{$x$, $\dot x$, $\tau^*$, $j^*$}
  $(x_p, S, Q) \gets $GetODEParam($j$, \texttt{cache}, $M$, $r$, $A$, $y$,
  \texttt{mean}=False) \;
  $a, b \gets$ ODECoef($x_p$, $S$, $x_0$, $\dot x_0$) \;
  $F_j, g_j, L_j \gets $ GetBoundariesForRegion \;
  $\tau^*, j^*, f, x, \dot x \gets$ EvolveToBoundary($\tau$, $x_p$, $a$, $b$, $F_j$,
  $g_j$, $L_j$, $j$) \;
  $Q_1 \gets Q[{:},1{:}d]$ \;
  $u_1 \gets f - Q_1 Q_1' f$ \;
  $(x_p, S, Q) \gets $GetODEParam($j^*$, \texttt{cache}, $M$, $r$, $A$, $y$, \texttt{mean}=False) \;
  $Q_2 \gets Q[{:}, 1{:}d]$ \;
  $pm \gets 1$ if $(\tau^* < t_{max})$ and $(j = j^*)$ else -1 \;
  $u_2 \gets pm (f - Q_2 Q_2' f)$ \;
  $V_1 \gets V_j(x)$ \;
  $V_2 \gets V_{j^*}(x)$ \;
  $\dot x \gets$ BoundaryDynamics($\dot x$, $u_1$, $u_2$, $V_1$, $V_2$) \;
 \end{algorithm}

\begin{algorithm}
  \caption{HMC1}
  \KwData{$N$, $t_{max}$, $j$, $x_0$, $\dot x_0$, $M$, $r$, $A$, $y$, $L$, $F$, $g$, \texttt{mean}}
  \KwResult{$x$, $\dot x$, $t_{remain}$}
  $n \gets \texttt{len}(x_0)$ \;
  $X \gets \texttt{array}((N,n))$ \;
  $\dot X \gets \texttt{array}((N,n))$ \;
  $R \gets \texttt{array}((N,))$ \;
  $x \gets x_0$; $\dot x \gets \dot x_0$,   $t \gets t_{max}$ \;
  \For{$i = 0:N$}{
    \While{$t > 0$}{
      $x$, $\dot x$, $\tau$, $j$ $\gets$ Evolve($t$, $j$, $x$, $\dot x$, $M$,
      $r$, $A$, $y$, $L$, $F$, $g$, \texttt{mean}) \;
      $t \gets t - \tau$ \;
    }
    $X[i,{:}] \gets x$ \;
    $\dot X[i,{:}] \gets \dot x$ \;
    $R[i] \gets j$ \;
    $t \gets t_{max}$ \;
    $\dot x \gets Null$ \;
  }
\end{algorithm}

\begin{algorithm}
  \caption{HMC2}
  \KwData{$N$, $t_{max}$, $j$, $x_0$, $\dot x_0$, $M$, $r$, $A$, $y$, $L$, $F$, $g$, \texttt{mean}}
  \KwResult{$x$, $\dot x$, $t_{remain}$}
  $n \gets \texttt{len}(x_0)$ \;
  $X \gets \texttt{array}((N,n))$ \;
  $\dot X \gets \texttt{array}((N,n))$ \;
  $R \gets \texttt{array}((N,))$ \;
  $I \gets \texttt{array}((N,))$ \;
  $x \gets x_0$; $\dot x \gets \dot x_0$,   $t \gets t_{max}$ \;
  \For{$i = 0:N$}{
    $x$, $\dot x$, $\tau$, $j$ $\gets$ Evolve($t$, $j$, $x$, $\dot x$, $M$, $r$, $A$, $y$, $L$, $F$,
    $g$, \texttt{mean}) \;
    $t \gets t - \tau$ \;
    $X[i,{:}] \gets x$ \;
    $\dot X[i,{:}] \gets \dot x$ \;
    $R[i] \gets j$ \;
    $I[i] \gets t = 0$\;
    \If{ $t = 0$ }{
      $t \gets t_{max}$ \;
      $\dot x \gets Null$ \;
    }
  }
\end{algorithm}

\end{document}